%% file: root.tex
\documentclass[letterpaper, 10pt, conference, margin=5pt]{ieeeconf}
\IEEEoverridecommandlockouts
\let\labelindent\relax
\usepackage{enumitem}
\usepackage[caption=false,font=footnotesize]{subfig}
\usepackage[utf8]{inputenc}
\usepackage[usenames,dvipsnames,table]{xcolor}
\usepackage{amsmath,amssymb}
\usepackage{hyperref}
\usepackage[noabbrev,capitalize]{cleveref}
\usepackage{accents}
\usepackage{graphicx}
\usepackage{pgfplots}
\usepackage{bcheader}
\usepackage{cite}
\usepackage{tikz}
\usetikzlibrary{calc}

\IEEEoverridecommandlockouts

\title{\LARGE \bf Stability of Gradient Learning Dynamics in Continuous Games: \\
Vector Action Spaces}%

\author{%
Benjamin J. Chasnov,  Daniel Calderone, Beh\c cet A\c c\i kme\c se, Samuel A. Burden, Lillian J. Ratliff%
\thanks{B. Chasnov, S. Burden, and L.J. Ratliff are with the Department of  Electrical and Computer Engineering, 
        University of Washington, Seattle, WA 98115
        {\tt\small $\{$bchasnov,sburden,ratliffl$\}$@uw.edu}}
        \thanks{D. Calderone and B. A\c c\i kme\c se are with the Department of Aeronautics and Astronautics,  University of Washington, Seattle, WA 98115
        {\tt\small $\{$djcal,behcet$\}$@uw.edu}}
\thanks{Funding for this work is provided by
NSF Award \#1836819
and NIH 5T90DA032436-09.}
}

\begin{document}

\maketitle
\thispagestyle{empty}
\pagestyle{empty}

\begin{abstract}
 \input{secs/0-abstract}
\end{abstract}

\section{Introduction}
\input{secs/1-intro}

\section{Preliminaries}
\label{sec:setup}

\input{secs/2-setup.tex}

\section{Stability in Zero-Sum and Potential Games}
\label{sec:scalingup}

\input{secs/3-decomposition}

\input{secs/3-zs-spectrum}

\input{secs/3-zs-stability}

\input{secs/3-zs-example}

\input{secs/3-pot-spectrum}

\input{secs/3-pot-stability}

\input{secs/3-pot-example}
\input{secs/3-pot-remark}

\subsection{Robustness to Variation in Time-scale Separation}
\input{secs/3-zs-nonuniform}

\input{secs/3-pot-nonuniform}

\begin{figure}[t]
\input{figs/block-spectrum}
\end{figure}




\section{Instability in General-Sum Games}
\label{sec:general-instability}
\input{secs/4-general-instability}
\input{secs/4-general-example}

\section{Numerical Example}
\label{sec:exampleslink}
\input{secs/5-numerical}

\section{Conclusion}

\input{secs/6-conclusion}

\begin{figure}[t]
\input{figs/potrot2}
\end{figure}

\bibliographystyle{plain}
\bibliography{refs}

\end{document}

%% file: secs/0-abstract.tex
Towards characterizing the optimization landscape of games,
this paper analyzes the stability of gradient-based dynamics near fixed points of two-player continuous games.
We introduce the quadratic numerical range as a method to characterize the spectrum of game dynamics and prove the robustness of equilibria to variations in learning rates.
By decomposing the game Jacobian into symmetric and skew-symmetric components, we 
assess the contribution of a vector field's potential and rotational components to the stability of differential Nash equilibria.
Our results show that in zero-sum games, all Nash are stable and robust; in potential games, all stable points are Nash.
For general-sum games, we provide a sufficient condition for instability.
We conclude with a numerical example in which learning with
timescale separation results in faster convergence.

%% file: secs/1-intro.tex
 The study of learning in games 
 is experiencing a resurgence in
 the control theory \cite{ratliff2016characterization,tang2019distributed,tatarenko2018learning},
 optimization \cite{mazumdar2018fundamental,mertikopoulos2019learning}, and 
 machine learning \cite{bu2019global,chasnov2019convergence,goodfellow2014gans,metz2016unrolled,fiez2019convergence} communities. 
 Partly driving this resurgence 
 is the prospect
 for game-theoretic
 analysis to yield
 machine learning algorithms
 that generalize better or are more robust.
A natural paradigm for learning in games is gradient play since updates in large decision spaces can be performed locally
while still guaranteeing local convergence in many problems~\cite{chasnov2019convergence,mertikopoulos2019learning}.

 Towards understanding the optimization landscape in such formulations, dynamical systems theory is emerging as a principal tool for analysis and ultimately synthesis~\cite{mazumdar2018fundamental,boone2019darwin,mertikopoulos2018cycles,berard2019closer,balduzzi2020smooth}.
One of the 
primary means to understand the 
optimization landscape 
of games is the eigenstructure and spectrum of the Jacobian of the learning dynamics in a neighborhood of a 
stationary point. However, as has been demonstrated~\cite{mazumdar2018fundamental}, not all attractors of the learning dynamics are game theoretically meaningful. Furthermore, structural heterogeneity in the learning algorithms employed by players can drastically change the convergence behavior.

\begin{figure}
\includegraphics[width=.49\linewidth]{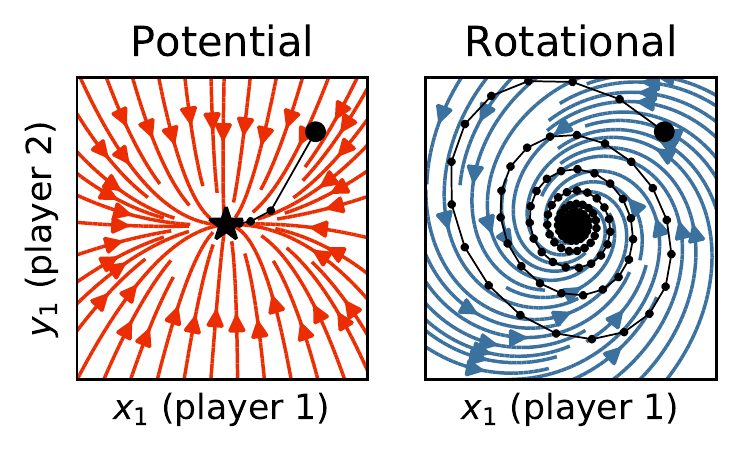}
\includegraphics[width=.49\linewidth]{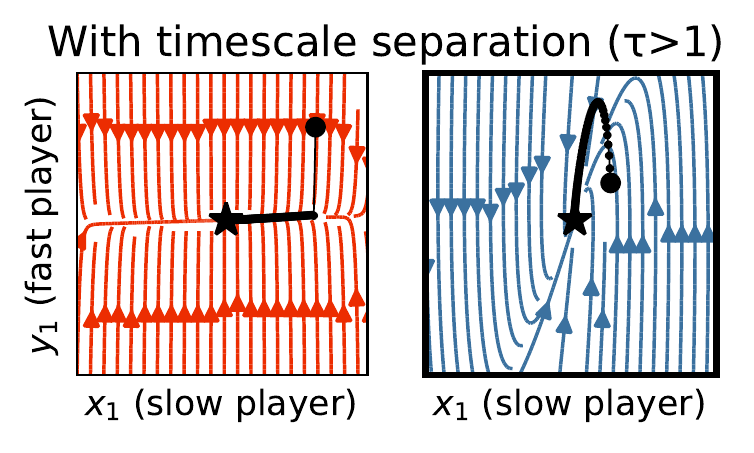}
\caption{\emph{Game dynamics with rotational components can converge at a faster rate with timescale separation.}
We plot slices of the
vector field and learning trajectories of mostly potential and mostly rotational learning dynamics from Example~\ref{ex:potrot-timescale}.
The game Jacobian at the equilibrium decomposes into $J=(1-\vep)S+\vep A$, where $S=S^\top$ is symmetric and $A=-A^\top$ is skew-symmetric. For the mostly potential system (red, $\vep=0.1$), players converge to the equilibrium without cycling.
For the mostly rotational system (blue, $\vep=0.9$), players without timescale separation cycle around the equilibrium.
Players with timescale separation take advantage of the rotational vector field to converge faster to the equilibrium, as shown in the right-most plot. See Fig.~\ref{fig:potrot} for a continuation of this example.}
\label{fig:potrot-vectorfields}
\end{figure}

The local stability of a hyperbolic 
 fixed point in a non-linear system can be assessed by examining the
 eigenstructure of the linearized dynamics~\cite{sastry1999nonlinear,khalil2002nonlinear}. However, in a game context extra structure comes from the underlying game---that is, players are constrained to move only along directions over which they have control. 
 They can only control their individual actions, as opposed to the entire state of the dynamical system corresponding to the learning rules being applied by the agents. 
 It has been observed in earlier work that not all stable attractors of gradient play are local Nash equilibria and not all local Nash equilibria are stable attractors of gradient play~\cite{mazumdar2018fundamental}. Furthermore, changes in players' learning rates---which corresponds to scaling rows of the Jacobian---can change an equilibrium from being stable to unstable and vice versa~\cite{chasnov2019convergence}.

To summarize, there is a subtle but extremely important difference between learning dynamics and traditional nonlinear dynamical systems:  alignment conditions are important for distinguishing between equilibria that have game-theoretic meaning versus those which are simply stable attractors of learning rules. Furthermore, features of learning dynamics such as learning rates can play an important role in shaping not only equilibria  but also alignment properties. Motivated by these observations, along with the recent resurgence of applications of learning in games in control, optimization, and machine learning, in this paper we provide an in-depth analysis of the spectral properties of gradient-based learning in two-player continuous games.

\textbf{Contributions:} Our main results are 
bounds on the spectrum of gradient-based learning dynamics near equilibria (\autoref{prop:zero-sum-spectrum}, \autoref{prop:pot-spectrum})
and robustness guarantees of differential Nash equilibria to variations in learning rates (\autoref{prop:zerosumrobust},
\autoref{prop:potentialrobust})
in two important classes of two-player continuous games: zero-sum (adversarial) and potential (implicitly cooperative). 
Moreover, we prove a sufficient condition for instability of learning dynamics (\autoref{thm:block-instability}). Finally, we include numerical examples  (Section~\ref{sec:exampleslink}) which provide further insights into the theoretical results.

More restrictive results applicable only in scalar action spaces were presented in an earlier conference paper~\cite{chasnov2020stability}. The present paper concerns the more general case of vector action spaces, introduces a novel decomposition of a general-sum game into its zero-sum and potential pieces, and applies a new analysis tool (quadratic numerical rage) to study stability of learning.


%% file: secs/2-setup.tex
This section contains game-theoretic preliminaries, mathematical formalism, and
a description of the gradient-based learning paradigm studied in this paper.

\subsection{Game-Theoretic Preliminaries}
A $2$-player \emph{continuous game} 
$\mc G = (\cost_1, \cost_2)$ 
is a pair of cost functions defined on a shared strategy space
$X=X_1\times X_2$
where player (agent)
$i \in \mc I=\{1,2\}$ 
has cost $f_i:X\to \mb{R}$. In this paper, the results apply to games with smooth costs $\cost_i\in C^r(
X,
\mb{R})$ for $r=2$.
Agent $i$'s set of feasible strategies 
is the $\dimm_i$-dimensional open and precompact set
$X_i \subseteq \mb R^{\dimm_i}$. 

The most common and arguably natural notion of an equilibrium in continuous games is due to  Nash~\cite{nash1951non}.
\begin{definition}[Local Nash equilibrium] A joint action profile $\fpx=(\fp{x_1},\fp{x_2})\in W_1\times W_2\subset X_1\times X_2$ is a local Nash equilibrium on $W_1\times W_2$ if, for each player $i\in\mc{I}$,
$f_i(\fp{x_i},\fp{x_{-i}})\leq f_i(x_i,\fp{x_{-i}})$, $\forall x_i \in W_i$.
\end{definition}
A local Nash equilibrium can equivalently be defined as in terms of best response maps: $\fp{x_i}\in \arg\min_{x_i}f_i(x_i,\fp{x_{-i}})$.  
From this perspective, local optimality conditions for players' optimization problems give rise to the notion of a differential Nash equilibrium~\cite{ratliff2013characterization,ratliff2016characterization}; non-degenerate differential Nash are known to be generic and structurally stable amongst local Nash equilibria in sufficiently smooth games~\cite{ratliff2014allerton}. Let $D_if_i$ denote the derivative of $f_i$ with respect to $x_i$ and, analogously, let $D_i(D_if_i)\equiv D_i^2f_i$ be player $i$'s individiaul Hessian.
\begin{definition}
\label{def:nash}
For continuous game $\mc{G}=(\cost_1, \cost_2)$ where $\cost_i \in C^2(X_1 \times X_2,\R)$,
 a joint action profile $(\fp{x_1},\fp{x_2}) \in X_1 \times X_2$ is a 
 \emph{differential Nash equilibrium}
 if 
 $D_if_i(\fp{x_1}, \fp{x_2}) = 0$ and $D_i^2f_i(\fp{x_1},\fp{x_2})> 0$ for each $i\in \mc{I}$.
\end{definition}
A differential Nash equilibrium is a strict local Nash equilibrium \cite[Thm.~1]{ratliff2013characterization}. Furthermore, the conditions $D_if_i(\fpx)=0$ and $D_i^2f_i(\fpx)\geq 0$ are necessary for a local Nash equilibrium \cite[Prop.~2]{ratliff2013characterization}.

Learning processes in games, and their study, arose as one of the explanations for how players grapple with one another in seeking an equilibrium~\cite{fudenberg1998theory}. In the case of sufficiently smooth games, gradient-based learning is a natural learning rule for myopic players.

\subsection{Gradient-based Learning as a Dynamical System %
}
At timestep $\kk\in\mathbb{N}$, 
a myopic agent $i$
updates its current action $x_i(\kk)$ by following the gradient of its
individual cost $\obj_i$
given the decisions of its competitors $x_{-i}$.
The synchronous adaptive process that arises is the discrete-time 
dynamical system
\begin{equation}
    x_i(t+1) =x_{i}(t)-\gamma_i D_if_i(x_i(t),x_{-i}(t))
    \label{eq:gradientplay}
\end{equation} for each $i\in \mc{I}$ where $D_if_i$ is the gradient of player $i$'s cost with respect to $x_i$ and $\gamma_i$ is player $i$'s learning rate. 

\subsubsection{Stability}
Recall that a matrix $A$ is called Hurwitz if its spectrum lies in the open left-half complex plane $\mb{C}_-^\circ$. Furthermore, we often say such a matrix is \emph{stable} in particular when $A$ corresponds to the dynamics of a linear system $\dot{x}=Ax$ or the linearization of a nonlinear system around a fixed point of the dynamics.\footnote{The Hartman-Grobman theorem~\cite{sastry1999nonlinear} states that around any hyperbolic fixed point of a nonlinear system, there is a neighborhood on which the nonlinear system is stable if the spectrum of the Jacobian lies in $\mb{C}_-^\circ$.}  

It is known that \eqref{eq:gradientplay} will converge locally asymptotically to a differential Nash equilibrium if the local linearization is a contraction~\cite{chasnov2019convergence}. 
Let 
\eqnn{g(x)=(D_1f_1(x),D_2f_2(x))
\label{eq:gameform}}
be the vector of individual gradients and let $Dg(x)$ be its Jacobian---i.e., the \emph{game Jacobian}. Further, let $\spec A\subset\mb{C}$
denote the \emph{spectrum} of the matrix $A$, and $\rho(A)$ its \emph{spectral radius}. 
Then, $\fpx$ is \emph{locally exponentially stable} if and only if $\rho(I-\gamma_1 \Lambda Dg(\fpx))<1$, where $\tau = \gamma_2/\gamma_1$ and $\Lambda=\blockdiag(I_{\dimm_1},\tau I_{\dimm_2})$
is a diagonal matrix and $I_{\dimm_i}$ is the identity matrix of dimension $\dimm_i$. 
The map $I-\gamma_1 \Lambda Dg(\fpx)$ is the local linearization of \eqref{eq:gradientplay}. 
Hence, to study stability (and, in turn, convergence) properties it is useful to analyze the spectrum of not only the map $I-\gamma_1 \Lambda Dg(\fpx)$ but also $\Lambda Dg(\fpx)$ itself. 

\subsubsection{Partitioning the Game Jacobian} 
\label{subsec:partition}
Let 
$\fixedpoint{x}=(\fixedpoint{x}_1,\fixedpoint{x}_2)$ 
be a joint action profile such that 
$g(\fixedpoint{x})=0$.
Towards better understanding the spectral properties of $Dg(\fixedpoint{x})$ (respectively, $\Learnrate Dg(\fixedpoint{x})$), we 
partition~$Dg(\fixedpoint{x})$ into blocks:
\eqnn{%
 J(\fpx)=
  \bmat{%
  -D_1^2\cost_1(\fpx) & 
  -D_{12}\cost_1(\fpx) \\
  -D_{21}\cost_2(\fpx) & 
  -D_{2}^2\cost_2(\fpx)
 }
 =
 \bmat{\Aa & \Bb \\
  \Cc & \Dd}.
 \label{eq:game-jacobian}
}
A differential Nash equilibrium (the second order conditions of which are sufficient for a local Nash equilibrium) is such that 
$\Aa<0$ and
$\Dd<0$. On the other hand, as noted above, $J$ is Hurwitz or stable if $\spec \jac \subset \mb{C}_-^\circ$. 
Moreover,
since the diagonal blocks are symmetric,
$\J$ is 
similar to the matrix 
in Fig~\ref{fig:blockdiag2p}.
For the remainder of the paper, 
we will study the $Dg$ at a given
fixed point $\fixedpoint{x}$ as defined in~\eqref{eq:game-jacobian}.
\begin{figure}[th]
  \begin{minipage}[c]{0.4\columnwidth}
    \input{figs/blockdiag2p.tex}
  \end{minipage}\hfill
  \begin{minipage}[c]{0.59\columnwidth}
    \caption{\emph{Similarity}: the game Jacobian in~\eqref{eq:game-jacobian} 
    is similar to a matrix with diagonal block-diagonals. The off-diagonals are arbitrary.
    }
  \label{fig:blockdiag2p}
  \end{minipage}
\end{figure}

\subsubsection{Classes of Games} 
Different classes of games can be characterized via $J$.   For instance,  a 
 \emph{zero-sum game}, where $\cost_1\equiv -\cost_2$, is such that $J_{12}=-J_{21}^\top$. On the other hand, %
 a game
 $\mc G = (\cost_1, \cost_2)$
 is a \emph{potential game} if and only if
 $D_{12}f_1 \equiv D_{21}f_2^\top$~\cite[Thm.~4.5]{monderer1996potential}, 
 which implies that $J_{12}=J_{21}^\top$.

\subsection{Spectrum of Block Matrices}
\label{sec:qnr}
One useful tool for characterizing the spectrum of a block operator matrix is the numerical range and quadratic numerical range, both of which 
contain the operator's spectrum~\cite{tretter2008spectral} and therefore all of its eigenvalues.
The \emph{numerical range} of $\J$ is defined by
\eqn{
   \NR(\J)=\{
\langle \J x,x\rangle:\ 
x\in \mb{C}^{\dimm_1+\dimm_2},\ %
\|x\|=1\} \subset \mb{C}, 
}
and is convex. 
Given a block operator $J$,
let
\begin{equation}
\J_{v,w}=\bmat{\langle \J_{11}v,v\rangle & \langle\J_{12}w,v\rangle\\ \langle\J_{21}v,w\rangle & \langle\J_{22}w,w\rangle} 
\label{eq:jvw}
\end{equation}
where $v\in \mb{C}^{\dimm_1}$ and $w\in \mb{C}^{\dimm_2}$.
The \emph{quadratic numerical range} of $\J$,
defined by
\eqnn{
\NR^2(\J)=\bigcup_{v\in \mc{S}_1, w\in \mc{S}_2 }\spec(\J_{v,w}),
\label{eq:qnr}}
is the union of the spectra of~\eqref{eq:jvw}
 where $\spec(\cdot)$ denotes the spectrum 
 of its argument 
 and $\mc{S}_i=\{v\in \mb{C}^{\dimm_i}:\ \|v\|=1\}$,
It is, in general, a non-convex subset of $\mb C$.
The quadratic numerical range~\eqref{eq:qnr}
is equivalent to the set of solutions of 
the characteristic polynomial
\eqnn{%
  \eig^2 
  &- \eig(\langle \Aa v,v\rangle  
  + \langle \Dd w,w\rangle ) 
  + \langle \Aa v,v\rangle\langle \Dd w,w\rangle  \\
  &- \langle \Bb v,w\rangle \langle \Cc w,v\rangle  = 0
  \label{eq:characteristic-polynomial}
}
for  $v\in \mc{S}_1$ and $w\in \mc{S}_2$. We use the notation $\langle \J x,y\rangle=\conjugate{x} \J y$ 
to denote the inner product. Note that $\NR^2(\J)$ is a subset of $\NR(\J)$ and, as previously noted, contains $\spec(\J)$.
Albeit non-convex, $\NR^2(\J)$ provides a tighter characterization of the spectrum.\footnote{There are numerous computational approaches for estimating the $\NR(\cdot)$ and $\NR^2(\cdot)$ (see, e.g., \cite[Sec. 6]{langer2001new}).}

Observing that the quadratic numerical range for a block $2\times 2$ matrix $J$ derived from a game on a finite dimensional Euclidean space reduces to characterizing the spectrum of $2\times 2$ matrices, 
we first characterize stability properties of scalar $2$-player continuous games.

%% file: figs/blockdiag2p.tex
	$
 J(\fixedpoint{\x},\fixedpoint{\y})
 \sim \left[\begin{tikzpicture}[scale=0.5,baseline={([yshift=-.5ex]current bounding box.center)}]

		\draw [ultra thick](0,2) -- (.7,1.3);
		\draw [ultra thick] (.7,1.3) -- (2,0);
		\fill [color=gray,color=gray,draw=white,very thick] 
		(0,0) rectangle (.7,1.3);
		\fill [color=gray,draw=white,very thick] 
		(.7,1.3) rectangle (2,2);
	\end{tikzpicture}\right]
 $

%% file: secs/3-decomposition.tex
In this section, we give stability results for
$2$-player continuous 
games on vector action spaces.
Consider a game $(f_1, f_2)$. Recall from the preliminaries that
$f_1,f_2\in C^2(X_1\times X_2,\mb{R})$ and
$X_1\subseteq \R^{d_1}, X_2\subseteq \R^{d_2}$,
are $d_i$-dimensional actions spaces.
Let $\fpx$ be a fixed point of~\eqref{eq:gameform} such that $g(\fpx)=0$.
We study the  gradient learning dynamics given in~\eqref{eq:gradientplay} near fixed points $\fpx$ by analyzing the spectral properties of the Jacobian of $g$.
\subsection{Jacobian Decomposition}
We decompose the game Jacobian,
\eqnn{%
\label{eq:blockdecomp}
\J(\fpx) =
\bmat{\Aa & \Pp \\ \Pp^\top & \Dd} + 
\bmat{0 & \Kk \\ -\Kk^\top & 0},
}
where $P=\frac{1}{2}(J_{12} + J_{21}^\top)$ and $Z=\frac{1}{2}(J_{12} - J_{21}^\top)$.  
As we will see, $P$ represents the potential-like part and $Z$ represents the zero-sum part.
The spectrum of $J(\fpx)$ is contained in the quadratic numerical range $\mc{W}^2(J(\fpx))$ which contains the spectrum of the matrices
\eqnn{
\label{eq:decomp:2x2}
J_{v,w} = \bmat{a & {p}+{z} \\ \conjugate{p}-\conjugate{z} & d}
}
where $a=\langle J_{11}v,v\rangle$, $d=\langle J_{22} w, w\rangle$ 
$p = \langle Pv,w \rangle$, and $z=\langle Z w, v\rangle$ constructed for unit-length complex numbers $v\in \mc{S}_1, w\in\mc{S}_2$.
Hence, to show the stability of a particular fixed point $\fpx$,
we must show that
for all $v\in \mc{S}_1, w\in\mc{S}_2$,
the spectrum of~\eqref{eq:decomp:2x2} 
is contained in the left-half complex plane.

%% file: secs/3-zs-spectrum.tex

For game $(f_1,f_2)$ with Jacobian $J(\fpx)$, define
the following for $i=1,2$: 
$\lambda_i^- =\min \spec (J_{ii}),\ 
\lambda_i^+ =\max \spec (J_{ii}).$
Additionally, define
\begin{alignat*}{3}
\lambda^- &=\min\{\lambda_1^-,\lambda_2^-\},\ 
\quad &&\underline{\lambda} &&=\tfrac{1}{2}(\lambda_1^-+\lambda_2^-),\\
\lambda^+ &=\max\{\lambda_1^+,\lambda_2^+\},\ 
&&\overline{\lambda} &&=\tfrac{1}{2}(\lambda_1^++\lambda_2^+).
\end{alignat*}
These terms depend on the spectrum of the players' individual Hessians and will be useful in deriving bounds on the spectrum of $J(\fpx)$.


\subsection{Zero-sum games ($P=0,Z=J_{12}$)}
The next theorem is our main result on zero-sum games, giving tight bounds on the spectrum of $J$ (i.e. bounds on the real and imaginary eigenvalues) near fixed points of the game.
Recall that for zero-sum game $(f, -f)$, the interaction term is $Z=-D_{12}f(x)$.
\begin{theorem}
[Spectrum of Zero-Sum Game Dynamics]
\label{prop:zero-sum-spectrum}
    Consider a zero-sum game $\mc{G}=(f,-f)$ and fixed point $\fpx$.
The Jacobian $J(\fpx)=-Dg(\fpx)$ of the dynamics
    $\dot{x}=-g(x)$ is such that
        \eqnn{\spec \jac \cap \mb{R}\subset \left[\lambda^-, \lambda^+\right]
        \label{eq:realeigens}}
    and $\spec \jac \backslash \mb{R}$ is contained in
    \begin{equation}
        \left\{z\in \mb{C}:\
        \mathrm{Re}(z)\in\left[\underline{\lambda}, \overline{\lambda}\,\right], \ |\mathrm{Im}(z)|\leq 
        \|Z\|\right\}
        \label{eq:imageigens}.
    \end{equation}
    Furthermore, if $\lambda_2^+<\lambda_1^-$ or $\lambda_1^+< \lambda_2^-$
    then the following two implications hold for
    $\delta=\lambda_1^--\lambda_2^+$ or $\delta=\lambda_2^--\lambda_1^+$,
    respectively:
(i) $\|Z
\|\leq \delta/2\ \implies\ \spec(J(x))\subset
            \mb{R}$;
(ii) $\|Z
\|>\delta/2\ \implies \
            \spec(J(x))\backslash \mb{R}\subset \{z\in \mb{C}: \ |\mathrm{Im}(z)|\leq
            \sqrt{\|Z
            \|^2-\delta^2/4}\}$.
\end{theorem}

\begin{proof} 
Observe that $\overline{\det(J_{v,w}(x)-\lambda I)}=\det(J_{v,w}(x)-\bar{\lambda}I)$ for
    $v\in\mc{S}_1$ and $w\in\mc{S}_2$
since $D_1^2f(x)$ and $-D_2^2f(x)$ are symmetric, which 
 implies
    that $\NR^2(J(x))=\NR^2(J(x))^\ast$.
Since $-w^\ast D_{12}f(x)^\top vv^\ast D_{12}f(x)w\leq 0$, 
    \eqref{eq:realeigens} and \eqref{eq:imageigens} follow from \cite[Prop.~1.2.6]{tretter2008spectral}, and 
    (i) and (ii) follow from \cite[Lem. 5.1-(ii)]{tretter2009spectral}.
\end{proof}

%% file: secs/3-zs-stability.tex
The following result shows that for zero-sum games, all differential Nash equilibria are stable under the gradient dynamics. 
\begin{corollary}[Stability in Zero-Sum Games]
\label{prop:zs-stable}
Consider a zero-sum game $\mc{G}=(f,-f)$ 
on finite dimensional action spaces $X_1,X_2$. 
If $x^\star$ is a differential Nash equilibrium of $\mc{G}$, then $x^\star$ is a locally stable equilibrium of $\dot x = -g(x)$.
\end{corollary}
\begin{proof}
From \autoref{prop:zero-sum-spectrum}, we have that the real parts of the spectrum of $J(\fpx)$ are upper-bounded by $\lambda^+$.
If $\fpx$ is a differential Nash equilibrium, then $ \lambda^+< 0$. Thus, $\fpx$ is a locally exponentially stable equilibrium. 
\end{proof}

While the Corollary above appears in ~\cite[Prop. 3.7]{mazumdar2018fundamental}, 
the novelty is showing that it is the special result of \autoref{prop:zero-sum-spectrum}.
%

\begin{remark}
\label{remark:zs:nonnash}
Zero-sum  games  can  have  stable  non-Nash equilibria. Players can get stuck at these spurious attractors of  the  learning  dynamics  where  the  individual  Hessians are not necessarily positive definite, i.e. players may converge to a point that is not a local minimum of their own cost.
\end{remark}


%% file: secs/3-zs-example.tex
The following applies \autoref{prop:zero-sum-spectrum} and is an example of Remark~\ref{remark:zs:nonnash}. 
\begin{example}
\label{ex:zs}
Consider the game $(f,-f)$ 
with cost
$f:\R^2\times\R^2 \to \R$
given by
\[\cost(x,y)=-x_1^2+3x_2^2-(2y_1^2+6y_2^2)+b(2y_1x_1+y_2x_2)\]
with $z\in\R$.
Direct computation shows that the origin is a stable equilibrium for $|b|>\sqrt{2}$, i.e. $b$ needs to be sufficiently large enough for the dynamics to be stable.
Moreover, by ~\autoref{prop:zero-sum-spectrum}, the imaginary parts of the spectrum of the game Jacobian are bounded by $\pm 2|b|$. 
This example demonstrates that 
%
interaction in zero-sum games can be necessary for stability.
%
\end{example}

%% file: secs/3-pot-spectrum.tex
\subsection{Potential games ($P=J_{12}, Z=0$)}
Recall that for potential games with potential function $\phi$,
the interaction term is $P=-D_{12}\phi(x)$. 
\begin{theorem}[Spectrum of Potential Game Dynamics]
\label{prop:pot-spectrum}
Consider a potential game $\mc G=(f_1, f_2)$. 
Let
\eqn{
\delta^\pm = \|P\|\tan\left(\frac{1}{2}\arctan \frac{2\|P\|}{|\lambda_1^\pm - \lambda_2^\pm|}\right).
}
The Jacobian $J(x)=-Dg(x)$ of the dynamics $\dot x= -g(x)$ at fixed points $\fpx$ is such that
$\spec J(x) \subset \R$
and
\begin{subequations}
\begin{align}
\lambda^-
- \delta^- 
&\leq \min \spec \jac
\leq 
\lambda^-
\label{eq:pot-spectrum1}
\\
\lambda^+
&\leq \max \spec \jac 
\leq \lambda^+
+ \delta^+.
\label{eq:pot-spectrum2}
\end{align}
\label{eq:pot-spectrum}
\end{subequations}
Furthermore, if $\lambda_2^+< \lambda_1^- $, then $\spec J(\fpx) \cap (\lambda_2^+, \lambda_1^-)$ is empty.
If $\lambda_1^+ < \lambda_2^-$, then $\spec J(\fpx) \cap (\lambda_1^-, \lambda_2^+)$ is empty.
\end{theorem}

\begin{proof}
Inequalities in \eqref{eq:pot-spectrum} 
follow from \cite[Prop. 1.2.4]{tretter2008spectral}
and last statements follow from \cite[Cor. 1.2.3]{tretter2008spectral}.
\end{proof}

%% file: secs/3-pot-stability.tex
The following result shows that for potential games,
all stable equilibria of the gradient dynamics are Nash. 
\begin{corollary}[Stability in Potential Games]
\label{proposition:potstable}
Consider a potential game $\mc{G}=(f_1,f_2)$ on finite dimensional action spaces $X_1,X_2$.
If $x^\star$ is a locally exponentially stable equilibrium of $\dot x = -g(x)$, then $x^\star$ is a differential Nash equilibrium of $\mc{G}$.
\end{corollary}
\begin{proof}
If $\fpx$ is stable, then $\max \spec J(x) < 0$. From \eqref{eq:pot-spectrum2} we have that $\max\{\lambda_1^+,\lambda_2^+\} < 0$. Therefore $\fpx$ is a differential Nash equilibrium.
\end{proof}


\begin{remark}
\label{remark:pot:nash}
Potential games can have unstable Nash equilibria.
That is, players can have local minimum of their costs which the gradient learning dynamics cannot converge to due to contribution of the interaction term. 
\end{remark}







%% file: secs/3-pot-example.tex
The next example applies \autoref{prop:pot-spectrum} and is an example of Remark~\ref{remark:pot:nash}.
\begin{example}
Consider the game $(f_1, f_2)$ with costs $f_i:\R^2\times\R^2\to\R,\ i=1,2,$ given by
\eqn{
f_1(x,y) &=\hphantom{3} x_1^2+2x_2^2 + p(x_1y_1 + x_2y_2),\\
f_2(x,y) &= 3y_1^2+4y_2^2 + p(x_1y_1 + x_2y_2)\\
}
with $p\in\R$.
Direct computation shows that the origin is unstable for $|p|>2\sqrt{3}$, 
i.e. in contrast to the zero-sum case in Example \ref{ex:zs}, larger interaction term causes instability.
By \autoref{prop:pot-spectrum}, we have that $\delta = p\tan(\arctan(p/2)/2)=\sqrt{p^2+4}-2$. Thus, the game Jacobian has eigenvalues that are in $[-8-\delta, -2+\delta]$. 
This example demonstrates that 
in potential games,
a lower bound on the interaction term can be necessary for stability.
\end{example}

%% file: secs/3-pot-remark.tex
\begin{remark}
In the setting of \autoref{prop:pot-spectrum}, we remark that if $J_{22}$ is invertible (without loss of generality), then 
the equilibrium is stable if and only if the Schur complement of $J(x)$ is negative, i.e. $\Aa - P\Dd^{-1}P^\top < 0$.
Corollary~\ref{proposition:potstable} ensures that this equilibrium is also a differential Nash equilibrium.
The proof of this statement is immediate from the properties of definite symmetric matrices (see, e.g.,~\cite{boyd2004convex}).
\end{remark}

%% file: secs/3-zs-nonuniform.tex
Recall 
that $\Lambda=\blockdiag(I_{d_1}, \tau I_{d_2})$ where $\tau$ is the learning rate ratio of the two players 
(Sec. \ref{subsec:partition}).
Below, we study the stability of $\dot x = -\Lambda g(x)$.
Our first result in this setting shows that differential Nash equilibria in zero-sum games are robust in variation in learning rates. 

\begin{theorem}[Robuesntess of Nash in Zero-sum Games] 
Consider a zero-sum game $\mc{G}=(f_1, f_2)=(f,-f)$ with game Jacobian $\J(x)=-\Lambda Dg(x)$. 
Suppose that $\fixedpoint{\x}$ is a differential Nash equilibrium.
%
Then, 
$\fixedpoint{\x}$ is a locally stable equilibrium of 
$\dot{x}=-\Lambda g(x)$ for any learning rate ratio $\tau$.
\label{prop:zerosumrobust}
\end{theorem}
\begin{proof}
First, observe that 
$a=\langle J_{11}v,v\rangle$ and $d=\langle J_{22} w, w\rangle$ are negative real numbers for any $v\in \mc{S}_1$ and $w\in \mc{S}_2$
by assumption that $\fixedpoint{\x}$
is a differential Nash equilibrium, i.e. $-D_i^2f_i(\fixedpoint{x}) < 0$ for each $i\in \mc \{1,2\}$.
Second, observe that for zero-sum games, $z=\langle Z w, v\rangle = -\conjugate{\langle Z v, w\rangle}$. 
Therefore, for $\fixedpoint{x}$ to be stable, the eigenvalues of 
\[
J_{v,w}=\bmat{ a &  z \\
-\tau \conjugate{z} & \tau d}
\]
must all be negative. 
Hence, we compute
the trace and determinant conditions  to be $\trace(J_{v,w})=\lambda_1 +\lambda_2=a + \tau d$ and $\det(J_{v,w})=\lambda_1 \lambda_2 = \tau(ad+|z|^2)$.
Notice that,
$\tau (ad+|z|^2) > 0 \iff ad+|z|^2 > 0$,
and $ a + \tau d < 0 \iff a + d < 0$.
Since $a,d<0$ and $\tau>0$,
both of the 
trace and determinant conditions for stability are satisfied, i.e. $\trace(J_{v,w}) < 0$ and $\det(J_{v,w}) > 0$.
Hence, $\fixedpoint{x}$ is a stable equilibrium of $\dot{x}=-\Lambda g(x)$.
\label{prop:zs-stable-lr}
\end{proof}

The stability of $\dot x=-\Lambda g(x)$ implies that there exists a range of learning rates $\gamma$ such that $x(\kk+1)=x(\kk)-\gamma\Lambda g(x(\kk))$ is locally asymptotically stable.

On the other hand,  differential Nash equilibria of potential games are not robust to variation in learning rates in general. However, the following theorem provides a sufficient condition that guarantees its robustness.

%% file: secs/3-pot-nonuniform.tex
\begin{theorem}[Robustness of Nash in Potential Games]
Consider a potential game $(f_1, f_2)$ with potential function $\phi$ and game Jacobian $J(x)=-\Lambda Dg(x)$.
Suppose $\fixedpoint{\x}$ is a differential Nash equilibrium. 
%
Then, 
$\fixedpoint{\x}$ is a locally stable equilibrium of $\dot{x}=-\Lambda g(x)$
for any learning rate ratio $\tau>0$ if 
$\lambda_1^-\lambda_2^- > \max \spec(P^\top P)$. 
\label{prop:potentialrobust}
\end{theorem}
\begin{proof}
First, observe that 
$a=\langle J_{11}v,v\rangle$ and $d=\langle J_{22} w, w\rangle$ are both negative real numbers for any $v\in \mc{S}_1$ and $w\in \mc{S}_2$
by assumption that $\fixedpoint{\x}$
is a differential Nash equilibrium, i.e. $-D_i^2f_i(\fixedpoint{x}) < 0$ for each $i\in \mc \{1,2\}$.
Second, observe that for potential games, $p=\langle P w, v\rangle = \overline{\langle P v, w\rangle}$. 
Therefore, for $\fixedpoint{x}$ to be stable, the eigenvalues of 
\[
J_{v,w}=\bmat{ a &  p \\
\tau \conjugate{p} & \tau d}
\]
must all have negative real components. 
Hence, we compute the 
the trace and determinant conditions  to be $\trace(J_{v,w})=\lambda_1 +\lambda_2=a + \tau d$ 
and $\det(J_{v,w})=\lambda_1 \lambda_2 = \tau(ad-|p|^2)$.
Notice that $ a + \tau d < 0 \iff a + d < 0$ and
$\tau (ad-|p|^2) > 0 \iff ad-|p|^2 > 0
\iff ad > |p|^2 > 0$. In terms of the original matrix, the condition
$\left(\max \spec(-J_{22})\right)\left( \max \spec(-J_{11})\right) > \max \spec(P^\top P)$ implies that $ad>|p|^2$ for all $v,w\in \mc{S}_1 \times \mc{S}_2$. Therefore, $\lambda_1^- \lambda_2^- > \max \spec(P^\top P)$ implies that $x$ is stable for all $\tau>0$.
\end{proof}

%% file: figs/block-spectrum.tex
\def\ticka{++(0,-.2) -- ++(0,.4) node[above,black]}
\def\tickb{++(0,.2) -- ++(0,-.4)node[below,black] }

\subfloat[Zero-sum game where $\delta=\lambda_2^- - \lambda_1^+ > 0$ and $\|J_{12}\|>\delta/2$ ]{
\begin{tikzpicture}[scale=0.8]
\draw[-stealth,gray] (-5,0) -- (5,0);	
\draw[-stealth,gray] (2,-1.5) -- (2,1.5);
\draw(1.8,-1) -- (2.2,-1);
\draw(1.8,1) -- (2.2,1) node [right]{\small $\sqrt{\|J_{12}\|^2-\delta^2/4}$};
\draw [gray] (2,0) node [above right] {$\mathbb C$};
\draw[very thick] (-4,0) -- (-2,0);
\draw[very thick,green!40!black] (-2,0) -- (-1,0);
\draw[very thick] (-1,0) -- (1,0);
\draw (-4,0) \tickb {$\lambda^-$};
\draw[green!40!black] (-2,0) \tickb {$\lambda_1^+$};
\draw[green!40!black] (-1,0) \tickb {$\lambda_2^-$};
\draw[green!40!black] (1,0) \tickb {$\lambda^+$};
\fill[opacity=.2,dashed] (-.5,-1) rectangle (-2.5,1);
\draw(-.5,0) \tickb {} node[below right] {$\overline{\lambda}$};
\draw(-2.5,0) \tickb {} node[below left] {$\underline{\lambda}$};
\draw[dashed] (-.5, -1) -- (2,-1);
\draw[dashed] (-.5, 1) -- (2,1);
\draw(-1.5,0)node[above,green!40!black]{$\delta$};
\draw (-3,0) node [above] {$W(J_{11})$};
\draw (0,0) node [above] {$W(J_{22})$};
\end{tikzpicture}}

\subfloat[Potential game where $\lambda_2^- - \lambda_1^+ > 0$.]{
\centering
\begin{tikzpicture}[scale=0.8]
	\draw[-stealth, gray] (-5,0) -- (5,0) node[above left] {$\mathbb R$};
	\draw  (-3,0) \tickb {$\lambda^-$};
	\draw[very thick] (-3,0) -- (-1,0);
	\draw[very thick] (1,0) -- (3,0);
	\draw (-1,0) \tickb {$\lambda_1^+$};
	\draw (1,0) \tickb {$\lambda_2^-$};
	\draw  (3,0) \tickb  {$\lambda^+$};
	\draw (3.5,0) node[above,black] {\color{green!30!black} $\delta_P^+$};
	\draw[very thick,green!40!black](3,0)--(4,0);
	\draw [green!40!black] (4,0) \ticka {};
	\draw (-3.5,0) node[above,green!40!black]{$\delta_P^-$};
	\draw [very thick,green!40!black](-4,0) -- (-3,0);
	\draw [green!40!black](-4,0)\ticka{};
	\draw  (-2,0) node[above]{$W(J_{11})$};
	\draw  (2,0) node[above]{$W(J_{22})$};
\end{tikzpicture}
}
\caption{\footnotesize 
\emph{Spectrum of learning dynamics near a fixed point in zero-sum and potential games.} 
We illustrate \autoref{prop:zero-sum-spectrum} (a, zero-sum game) and \autoref{prop:pot-spectrum} (b, potential game).
The highlighted and thick regions contain the spectrum of the linearized dynamics.
}

%% file: secs/4-general-instability.tex
As a complementary result to the stability analysis for zero-sum and potential games, we provide a sufficient condition for the \emph{in}stability of fixed points of gradient-based learning dynamics in general sum games.
Our results quantifies the contribution of 
the off-diagonal interaction terms of~\eqref{eq:game-jacobian} in destabilizing equilibria.

We begin by 
expressing the game Jacobian as the sum of symmetric and skew-symmetric matrices, $J=\tfrac{1}{2}(J+J^\top) + \tfrac{1}{2}(J-J^\top)$.
Let $R$ be a rotation that diagonalizes $\frac{1}{2}(J+J^\top)$ and sorts the eigenvalues so that
$\J$ decomposes into
\eqnn{%
RJR^\top = \bmat{M_+ & 0 \\ 0 & M_-} + \bmat{Z_{1} & Z_2 \\ - Z_2^\top & Z_3 }
\label{eq:sort}
}
where $M_+> 0$, $M_-\le 0$ are diagonal and $Z_1$ and $Z_3$ are skew-symmetric. Let $\lambda^-(M_+) > 0$ be the minimum eigenvalue of $M_+$ and $\lambda^+(M_-)\leq 0$ be the maximum eigenvalue of $M_-$.  
\begin{theorem}[Sufficient Conditions for 
Instability in 
Games]
\label{thm:block-instability}
Consider general-sum game $\mc{G}=(\cost_1,\cost_2)$ with $f_i\in C^2(X_1\times X_2,\mb{R})$ where $X_i$  is $\dimm_i$-dimensional for each $i=1,2$.
At a fixed point $\fixedpoint{x}$, $\spec J(\fixedpoint{x})\not\subset \mb{C}_-^\circ$ if
\eqnn{
\|Z_2 \| < \tfrac{1}{2}\big(|\lambda^+(M_-)|+|\lambda^-(M_+)|\big) <
|\lambda^-(M_+)| 
\label{eq:gsunstable}
}
with $M_+, M_-$ and $Z_2$ defined in~\eqref{eq:sort}.
\end{theorem}
\begin{proof}
Since $Z_1$ and $Z_3$ are skew-symmetric we have that  $\text{Re}\big(M_- + Z_{3})\big)\leq \lambda^+(M_-) \leq 0$ and $0 \leq \lambda^-(M_+) \leq \text{Re}\big(W(M_+ + Z_{1})\big)$~\cite[Prop. 1.1.12]{tretter2008spectral}.  
\end{proof}

The result above works by bounding 
a non-empty subset of the eigenvalues of $J$ in $\mb{C}_+^\circ$ to guarantee instability. 
%
The inequalities in~\eqref{eq:gsunstable} 
are the block matrix equivalent of being 
inside the circle of radius $\sqrt{h^2 + p^2}$ in the scalar case~\cite{chasnov2020stability}.


%% file: secs/4-general-example.tex
\begin{example}
Consider a game $(f_1, f_2)$ with costs
$f_i:\R^2\times\R^2\to\R$, $i=1,2$ given by
\eqn{
f_1(x, y) &= -x_1^2+3x_2^2-x_1x_2 + 8x_1y_1 -2x_2y_2,\\
f_2(x, y) &= \hphantom{-}y_1^2+4y_2^2 -y_1y_2 + 2 x_1 y_2 + 2 x_2y_1.
}
By diagonalizing the symmetric component of the game Jacobian, we have that $M_+ = 4.8$, $M_-=-\diag(4.4,5.7,8.7)$. By applying this rotation to the skew-symmetric component, we have that $\|Z_2\| = 4.0$ using the Euclidean norm.
Since $\|Z_2\| < 4.6 < 4.8$,  we have that the origin is unstable.
\end{example}


%% file: secs/5-numerical.tex
\label{sec:ex}

\begin{example}
\label{ex:potrot-timescale}
We 
explore
how timescale separation can improve the convergence of game dynamics. 
In particular, we show that when the vector field has enough rotational component, 
timescale separation can lead to 
a well-conditioned Jacobian and thus
faster convergence.
Consider a zero-sum game $\mc G=(f,-f)$ on $\R^2\times\R^2$ with cost given by
\eqn{
f(x,y)= (1-\vep) \left(x_1^2+\tfrac{3}{2}x_2^2-2y_1^2-\tfrac{5}{2}y_2^2\right)+\vep x^\top B y
}
and the matrix $B$ is such that each entry is $B_{ij}=1$ for each $i,j$ except for $B_{22}=-1$.
The parameter $0\leq \vep \leq 1$ controls the amount of rotation in the game vector field.
When $\vep=0$, the game Jacobian is symmetric; when $\vep=1$, the game Jacobian is skew-symmetric. The decomposition of the Jacobian is $J=(1-\vep )S+\vep A$ where $S=S^\top$ and $A=-A^\top$.
Suppose agents descend their individual gradient with learning rates $\gamma_1,\ \gamma_2=\tau\gamma_1$, yielding discrete-time dynamics
\begin{equation}
    \begin{split}
        x(t+1) &= x(t) - \gamma_1 D_1f(x(t),y(t))\\
y(t+1) &= y(t) + \gamma_1 \tau  D_2f(x(t), y(t)).
    \end{split}
    \label{ex:potrot}
\end{equation}
We initialize $x(0), y(0)$ to a vector of ones and simulate~\eqref{ex:potrot} with $\gamma_1=10^{-3}$. We plot the 2-norm of the iterates in Fig.~\ref{fig:potrot:a}.

Recall that the spectrum of $\Lambda J(x,y)$ at an equilibrium $(x,y)$ determines its stability and that the spectral radius of $I+\gamma\Lambda J(x,y)$ determines the  
convergence rate of the discrete-time system above, where $\Lambda = \text{blockdiag}(I_1,\tau I_2)$. 
These quantities with varying $\tau>0$ are plotted in Fig.~\ref{fig:potrot:b}.

By learning with different rates $\gamma_i$, players can take advantage of the rotational component of a vector field to converge faster. 
For $\vep=0.9$, the system converges fastest with $\tau\approx28$. This is indicated by the blue curves in Fig.~\ref{fig:potrot:a}, black squares in Fig.~\ref{fig:potrot:b}, and the right plot in Fig.~\ref{fig:potrot-vectorfields}.

\end{example}

%% file: secs/6-conclusion.tex
We 
characterize
local stability 
of Nash equilibria 
in two-player games
by analyzing the spectrum 
of the
gradient learning dynamics.
By decomposing the game Jacobian into zero-sum and potential game components,
we assess how each term contributions
to the stability of Nash or non-Nash equilibria.
We provide tight bounds on the spectrum of the learning dynamics near fixed-points. 
%
Such results give valuable insights 
into the interaction of algorithms and optimization landscape of settings
most accurately 
modeled 
as games.

In the numerical example, we demonstrate an important trade-off between timescale separation between agents and the skew-symmetric component of the learning dynamics.
Agents learning at different rates can result in faster convergence 
if the vector field has enough rotational component.
This indicates a future direction of research: to optimize convergence rate 
given the strength of the skew-symmetric component of the game.

%% file: figs/potrot2.tex
\centering



\subfloat[
The rotational system (blue) with timescale separation (right) achieves the fastest convergence by taking advantage of the rotational vector field.\label{fig:potrot:a}
]{\includegraphics[width=1\linewidth]{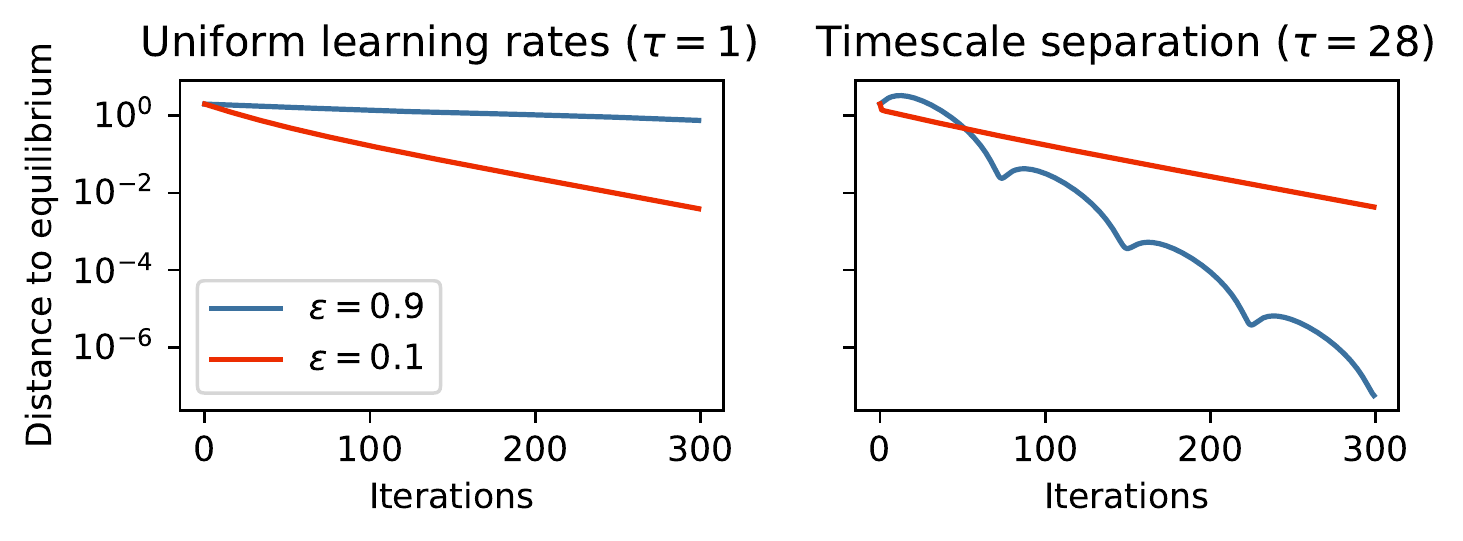}}



\subfloat[The spectral radius of $I+\gamma_1 \Lambda J(z)$ for the discrete-time update 
and the eigenvalues of $\Lambda J(z)$ for the continuous-time system $\dot z = -\Lambda g(z)$ at equilibrium $z=(x,y)=0$ for increasing learning rate ratio $\tau>0$.
\label{fig:potrot:b}]
{\includegraphics[height=.35\linewidth]{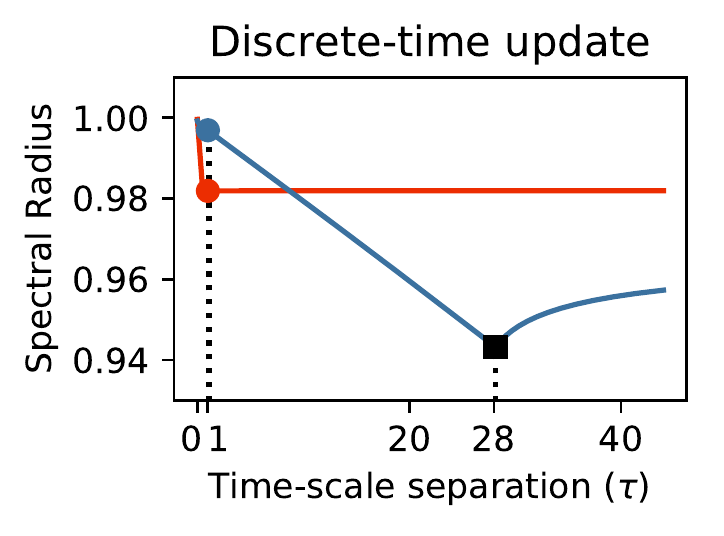}
\includegraphics[height=.35\linewidth]{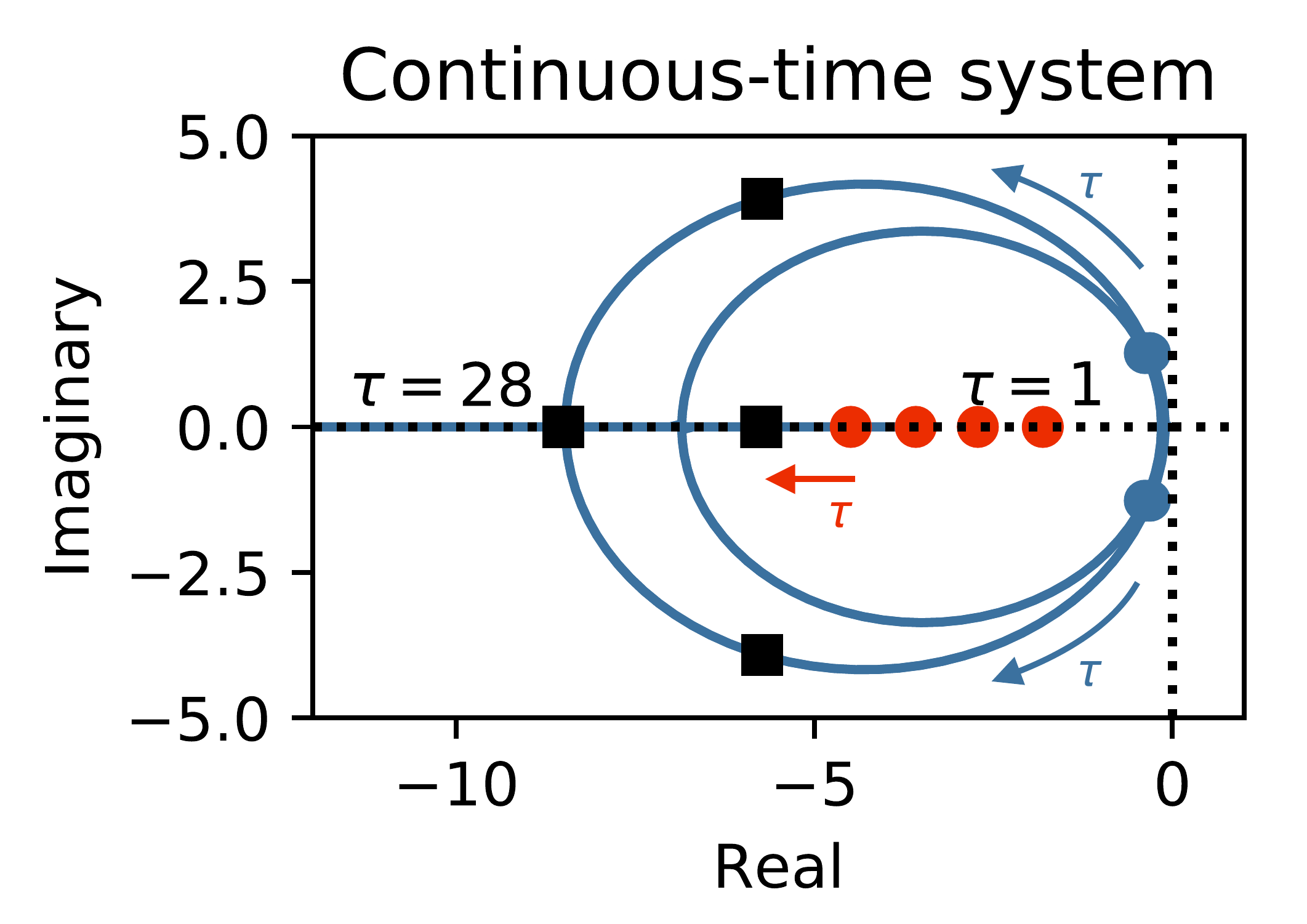}}




\caption{\emph{Faster convergence  with timescale separation.}
(a) Timescale separation improves convergence rates of systems with mostly rotational dynamics (blue).
(b) The spectral radius and spectrum of the discrete-time and continuous-time updates, respectively, show that at $\tau\approx28$, the mostly-rotational system 
achieves fastest convergence because it takes advantage of the imaginary eigenvalues to achieve a smaller spectral radius.
}
\label{fig:potrot}

%% file: root.bbl
\begin{thebibliography}{10}

\bibitem{balduzzi2020smooth}
David Balduzzi, Wojiech~M Czarnecki, Thomas~W Anthony, Ian~M Gemp, Edward
  Hughes, Joel~Z Leibo, Georgios Piliouras, and Thore Graepel.
\newblock {Smooth markets: A basic mechanism for organizing gradient-based
  learners}.
\newblock {\em Proc.~Inter.~Conf.~Learning Representations}, 2020.

\bibitem{berard2019closer}
Hugo Berard, Gauthier Gidel, Amjad Almahairi, Pascal Vincent, and Simon
  Lacoste-Julien.
\newblock {A closer look at the optimization landscapes of generative
  adversarial networks}.
\newblock {\em Proc.~Inter.~Conf.~Learning Representations}, 2020.

\bibitem{boone2019darwin}
Victor Boone and Georgios Piliouras.
\newblock {From Darwin to Poincar{\'e} and von Neumann: Recurrence and Cycles
  in Evolutionary and Algorithmic Game Theory}.
\newblock In {\em Inter.~Conf.~Web and Internet Economics}, pages 85--99, 2019.

\bibitem{boyd2004convex}
S.~Boyd and L.~Vandenberghe.
\newblock {\em Convex Optimization}.
\newblock Cambridge Univ. Press, 2004.

\bibitem{bu2019global}
Jingjing Bu, Lillian~J Ratliff, and Mehran Mesbahi.
\newblock Global convergence of policy gradient for sequential zero-sum linear
  quadratic dynamic games.
\newblock {\em arXiv preprint arXiv:1911.04672}, 2019.

\bibitem{chasnov2020stability}
Benjamin Chasnov, Dan Calderone, Beh{\c{c}}et A{\c{c}}{\i}kme{\c{s}}e, Samuel~A
  Burden, and Lillian~J Ratliff.
\newblock Stability of gradient learning dynamics in continuous games: Scalar
  action spaces.
\newblock In {\em IEEE Conf. on Decision and Control}, December 2020.

\bibitem{chasnov2019convergence}
Benjamin Chasnov, Lillian Ratliff, Eric Mazumdar, and Samuel Burden.
\newblock {Convergence Analysis of Gradient-Based Learning in Continuous
  Games}.
\newblock In {\em Proc.~Uncertainty in Artificial Intelligence}, 2019.

\bibitem{fiez2019convergence}
Tanner Fiez, Benjamin Chasnov, and Lillian~J Ratliff.
\newblock {Implicit Learning Dynamics in Stackelberg Games: Equilibria
  Characterization, Convergence Analysis, and Empirical Study}.
\newblock {\em Proc.~Inter.~Conf.~Machine Learning}, 2020.

\bibitem{fudenberg1998theory}
Drew Fudenberg and David~K Levine.
\newblock {\em The theory of learning in games}.
\newblock MIT press, 1998.

\bibitem{goodfellow2014gans}
Ian~J. Goodfellow, Jean Pouget-Abadie, Mehdi Mirza, Bing Xu, David
  Warde-Farley, Sherjil Ozair, Aaron Courville, and Yoshua Bengio.
\newblock {Generative Adversarial Nets}.
\newblock In {\em Advances in Neural Information Processing Systems}, 2014.

\bibitem{khalil2002nonlinear}
Hassan~K Khalil.
\newblock {\em Nonlinear systems theory}.
\newblock Prentice Hall, 2002.

\bibitem{langer2001new}
Heinz Langer, A~Markus, V~Matsaev, and C~Tretter.
\newblock A new concept for block operator matrices: the quadratic numerical
  range.
\newblock {\em Linear algebra and its applications}, 330(1-3):89--112, 2001.

\bibitem{mazumdar2018fundamental}
Eric Mazumdar, Lillian~J Ratliff, and Shankar Sastry.
\newblock On gradient-based learning in continuous games.
\newblock {\em SIAM Journal on Mathematics of Data Science}, 2(1):103--131,
  2020.

\bibitem{mertikopoulos2018cycles}
Panayotis Mertikopoulos, Christos Papadimitriou, and Georgios Piliouras.
\newblock Cycles in adversarial regularized learning.
\newblock In {\em Proc.~29th Ann.~ACM-SIAM Symp.~Discrete Algorithms}, pages
  2703--2717. SIAM, 2018.

\bibitem{mertikopoulos2019learning}
Panayotis Mertikopoulos and Zhengyuan Zhou.
\newblock Learning in games with continuous action sets and unknown payoff
  functions.
\newblock {\em Mathematical Programming}, 173(1-2):465--507, 2019.

\bibitem{metz2016unrolled}
Luke Metz, Ben Poole, David Pfau, and Jascha Sohl-Dickstein.
\newblock Unrolled generative adversarial networks.
\newblock {\em Proc.~Inter.~Conf.~Learning Representations}, 2017.

\bibitem{monderer1996potential}
Dov Monderer and Lloyd~S Shapley.
\newblock Potential games.
\newblock {\em Games and economic behavior}, 14(1):124--143, 1996.

\bibitem{nash1951non}
John Nash.
\newblock Non-cooperative games.
\newblock {\em Annals of mathematics}, pages 286--295, 1951.

\bibitem{ratliff2014allerton}
L.~J. {Ratliff}, S.~A. {Burden}, and S.~S. {Sastry}.
\newblock {Genericity and structural stability of non-degenerate differential
  Nash equilibria}.
\newblock In {\em Proc.~Amer.~Control~Conf.}, pages 3990--3995, 2014.

\bibitem{ratliff2013characterization}
Lillian~J Ratliff, Samuel~A Burden, and S~Shankar Sastry.
\newblock {Characterization and computation of local Nash equilibria in
  continuous games}.
\newblock In {\em Proc.~51st Ann.~Allerton Conf.~Communication, Control, and
  Computing}, pages 917--924. IEEE, 2013.

\bibitem{ratliff2016characterization}
Lillian~J Ratliff, Samuel~A Burden, and S~Shankar Sastry.
\newblock {On the Characterization of Local {Nash} Equilibria in Continuous
  Games}.
\newblock {\em IEEE Trans~Automa.~Control}, 61(8):2301--2307, 2016.

\bibitem{sastry1999nonlinear}
S.~Shankar Sastry.
\newblock {\em Nonlinear systems: analysis, stability, and control}.
\newblock Springer-Verlag New York, 1999.

\bibitem{tang2019distributed}
Yujie Tang and Na~Li.
\newblock Distributed zero-order algorithms for nonconvex multi-agent
  optimization.
\newblock In {\em Proc.~57th Allerton Conf. Communication, Control, and
  Computing}, pages 781--786, 2019.

\bibitem{tatarenko2018learning}
T.~{Tatarenko} and M.~{Kamgarpour}.
\newblock {Learning Nash Equilibria in Monotone Games}.
\newblock In {\em Proc.~IEEE Conf.~Decision and Control}, pages 3104--3109,
  2019.

\bibitem{tretter2008spectral}
Christiane Tretter.
\newblock {\em Spectral theory of block operator matrices and applications}.
\newblock World Scientific, 2008.

\bibitem{tretter2009spectral}
Christiane Tretter.
\newblock Spectral inclusion for unbounded block operator matrices.
\newblock {\em J.~functional analysis}, 256(11):3806--3829, 2009.

\end{thebibliography}
